\documentclass[sn-mathphys-num]{sn-jnl}

\usepackage{booktabs}
\usepackage{amsmath,amssymb,amsfonts}
\usepackage{amsthm}
\usepackage[ruled]{algorithm2e}
\usepackage{xcolor}

\newcommand{\dataset}[1]{\texttt{#1}}
\newcommand{\algo}[1]{\textsc{#1}}
\newcommand{\Chen}{\algo{ChenEtAl}\xspace}
\newcommand{\Jones}{\algo{Jones}\xspace}
\newcommand{\Ours}{\algo{Ours}\xspace}
\newcommand{\OursObl}{\algo{OursOblivious}\xspace}

\DeclareMathOperator*{\argmin}{arg\,min}
\renewcommand{\epsilon}{\varepsilon}
\newcommand{\BO}[1]{O\left(#1\right)}

\newcommand{\Let}[2]{#1 $\leftarrow$ #2}

\newcommand{\M}{\ensuremath{\mathcal{M}}}
\newcommand{\X}{\ensuremath{\mathcal{X}}}

\newcommand{\rank}[1]{\operatorname{rank}{\left(#1\right)}}

\newcommand{\OPT}{\mbox{\rm OPT}}
\newcommand{\TTL}{\mbox{\rm TTL}}
\newcommand{\dmin}{d_{\rm min}}
\newcommand{\dmax}{d_{\rm max}}
\newcommand{\AVg}{\ensuremath{AV_\gamma}}
\newcommand{\RVg}{\ensuremath{RV_\gamma}}

\newcommand{\Ag}{\ensuremath{A_\gamma}}
\newcommand{\Rg}{\ensuremath{R_\gamma}}
\newcommand{\Og}{\ensuremath{O_\gamma}}

\newcommand{\repVg}{\ensuremath{repV_\gamma}}
\newcommand{\repsCg}{\ensuremath{repsC_\gamma}}

\newcommand{\AVgt}[1][]{\ensuremath{AV_{\gamma, #1}}}
\newcommand{\RVgt}[1][]{\ensuremath{RV_{\gamma, #1}}}

\newcommand{\Agt}[1][]{\ensuremath{A_{\gamma, #1}}}
\newcommand{\Rgt}[1][]{\ensuremath{R_{\gamma, #1}}}

\newcommand{\repsCgt}[1][]{\ensuremath{repsC_{\gamma, #1}}}

\newtheorem{theorem}{Theorem}
\newtheorem{corollary}{Corollary}
\newtheorem{fact}{Fact}
\newtheorem{lemma}{Lemma}

\begin{document}

\title{Fair Center Clustering in Sliding Windows}

\author*[1]{\fnm{Matteo} \sur{Ceccarello}}\email{matteo.ceccarello@unipd.it} 
\author[1]{\fnm{Andrea} \sur{Pietracaprina}}\email{andrea.pietracaprina@unipd.it} 
\author[1]{\fnm{Geppino} \sur{Pucci}}\email{geppino.pucci@unipd.it} 
\author[1]{\fnm{Francesco} \sur{Visonà}}\email{francesco.visona.3@studenti.unipd.it}

\affil[1]{\orgdiv{Department of Information Engineering}, \orgname{University of Padova}, \country{Italy}}

\abstract{
The $k$-center problem requires the selection of $k$ points (centers) from a given metric pointset $W$ so to minimize the maximum distance of any point of $W$ from the closest center. This paper focuses on a fair variant of the problem, known as \emph {fair center}, where each input point belongs to some category and each category may contribute a limited number of points to the center set. We present the first space-efficient streaming algorithm for fair center in general metrics, under the sliding window model. At any time $t$, the algorithm is
able to provide a solution for the current window whose quality is almost as good as the one
guaranteed by the best, polynomial-time sequential algorithms run on the entire window,
and exhibits space and time requirements independent of the 
window size. Our theoretical results are backed by an extensive set of experiments on both real-world and synthetic datasets, which provide evidence of the practical viability of the algorithm.
}


\maketitle

\section{Introduction}
Clustering  is a fundamental primitive for machine learning and data mining, with applications in many domains \citep{HennigMMR15}. One popular variant is $k$-center clustering, which has also been intensely studied in the realm of facility location \citep{Snyder11}. Given a set of points $W$ from a metric space and an integer $k \leq |W|$, the $k$-center problem requires to select $k$ points (dubbed \emph{centers}) from $W$,  which minimize the maximum distance of any point of $W$ from its closest center. The centers induce immediately a partition of  $W$ into $k$ \emph{clusters}, one per center,  where each point is assigned to the cluster associated to its closest center. When the distance between points models (dis)similarity,  each center can then be regarded as a suitable representative for all the points in the corresponding cluster. 
Indeed, efficient clustering approaches have been traditionally used for two main purposes, one being the partitioning of the dataset into groups  of similar data points for unsupervised classification, the other being the selection of a small number of representative points (the centers of the clusters) that are good descriptors of the entire dataset. In particular, 
$k$-center has been often  employed  as a summarization primitive to
extract succinct coresets from large datasets, where computationally expensive analyses can then be performed (e.g., see \citep{CeccarelloPP20,DBLP:journals/jbd/CeccarelloPPS23} and references therein). 

In many applications, the points being clustered represent features of \emph{individuals}, and the result of the clustering informs decisions that may have an impact on these individuals' livelihoods.
In order not to discriminate some groups of individuals with respect to some specific, sensible attribute (e.g., ethnicity, gender, or political views), when clustering is used for summarization purposes, it is of utmost importance that the centers provide a \emph{fair} representation of the population with respect to this attribute.
To this end, in \emph{fair center clustering}, we label each point with a \emph{color}, representing the attribute value for the group which the corresponding individual belongs to, and we impose an upper limit on the number of returned centers from each color, so that no color, hence no group, is over-represented in the output.
We remark that blindly ignoring sensible attributes does not imply a fair solution~\cite{DBLP:conf/innovations/DworkHPRZ12}: in fact, group membership information can leak from other features of the points, or the features themselves might be discriminative towards some group.

Scalable versions of this fair $k$-center variant
have recently been studied in distributed and insertion-only streaming settings \citep{DBLP:journals/jbd/CeccarelloPPS23}. In this work, we devise and experiment with
the first efficient streaming algorithm for fair center under the more challenging, heavily studied sliding-windows model \citep{DatarM16}, where the sought solution at any time must refer only to a fixed-size window $W$ of the most recent data, disregarding older elements of the stream, so to account for concept drift in the data. 

\subsection{Related work}

The $k$-center problem, being a fundamental primitive in many data analysis workflows, has been studied extensively in the past decades. For brevity, in this section we give an account of those results that relate most closely to our work. 

In the sliding window model,
\citep{DBLP:conf/icalp/Cohen-AddadSS16,PellizzoniPP22}
provide approximation algorithms for the unrestricted $k$-center problem without fairness
constraints. The variant of $k$-center with outliers in the sliding window model has been studied in
\citep{PellizzoniPP22a}.
A relevant generalization of the $k$-center problem is the \emph{matroid center} problem,
where the centers of the clustering are required to be an independent set of
a given matroid.
The seminal work of
\citep{ChenLLW16} 
provided the first 3-approximation sequential algorithm for this problem, albeit featuring a high computational complexity. 
The matroid center problem has also been studied in
\citep{DBLP:conf/approx/Kale19, DBLP:journals/jbd/CeccarelloPPS23} in the insertion-only streaming  setting, yielding a $(3+\epsilon)$ approximation and also allowing for the presence of outliers.

The fair center problem studied in this paper can be seen as a specialization of the matroid center problem for the so-called \emph{partition matroid}, built on points belonging to   different categories and whose  independent sets are those containing no more than  a fixed upper limit of  points per category. In 
\citep{DBLP:conf/icml/KleindessnerAM19}  the authors provide a sequential $(3\cdot 2^{\ell-1} -1)$-approximation   algorithm for fair center (where $\ell$ is the number of categories),
with a runtime linear both in $k$ and in the number of points.
Later,
\citep{DBLP:conf/icml/JonesNN20} 
improved the approximation factor to 3 while retaining the same time complexity.  \citep{DBLP:conf/aistats/Amagata24} studies the fair clustering problem in the presence of outliers, providing a sequential, randomized bicriteria approximation  algorithm.
A further specialization of the problem, where each center is required to represent at least a given fraction of the input points, in addition to respecting the fairness constraints, has been studied by
\citep{DBLP:conf/icml/AngelidakisKSZ22}.
In the insertion-only streaming setting, the fair center problem has
been considered in 
\citep{DBLP:conf/icml/ChiplunkarKR20,DBLP:journals/corr/GanGYZ23,DBLP:conf/pakdd/LinGJ24}. These papers all yield 3 approximations for the problem, with improvements in the working space requirements over  the general matroid center algorithms discussed before.  To the best of our knowledge, ours is the first work to address the fair center problem in the more challenging sliding window setting.

We wish to stress that, although widely studied, the notion of fairness in center selection adopted in this paper is not the only possible one.
For instance, 
\citep{JungKL19, DBLP:conf/icml/MahabadiV20} 
consider \emph{individually fair} clustering, where each point is required
to have at least one center among its $n/k$ nearest neighbors.
\citep{DBLP:journals/jgo/HanXXY23} also consider this problem in the presence of outliers.
Also,  observe that all the approaches discussed so far refer to enforcing fairness conditions on the centers, and are thus tightly associated to the notion of clustering as a summarization primitive.  There is a very prolific line of research  which regards fairness as a set of balancing constraints on the elements of the partitioning induced by the centers (e.g., the clusters).  For an account of recent developments in this direction (not explored in this paper), we refer the
interested reader to \citep{DBLP:conf/www/CeccarelloPP24} and references
therein.

\subsection{Our contribution}
We present the first sliding window algorithm for fair center clustering in general metric spaces, which, at any time $t$, is 
able to provide an accurate solution for the current window,  and requires space and time independent of the 
window size.  
More specifically, let $S$ denote a potentially infinite stream of points from a metric space, and let $n>0$ be the target window size. Each point of $S$ is associated with one of $\ell$ colors and fairness is modeled by requiring that any feasible solution to the problem contains
at most $k_i$ points of color $i$, for every $1 \leq i \leq \ell$.
Let $\epsilon \in (0,1)$ be a fixed accuracy parameter, $\Delta$ be the aspect ratio of $S$ (i.e., the ratio between maximum and minimum pairwise distance), and $D_{W}$ be the doubling dimension of a window $W \subset S$ (formally defined in Section~\ref{sec:prelim}), which 
generalizes the notion of
Euclidean dimensionality to general metrics.
Also,  define $k = \sum_{i=1}^{\ell} k_i$ and 
let $\alpha$ denote the best approximation ratio guaranteed by a polynomial-time sequential algorithm for fair center (currently $\alpha=3$ \citep{DBLP:conf/icml/JonesNN20}). Our main contributions are the following (detailed statements are found in Section~\ref{sec-fairkcenter}).

\begin{itemize}
\item 
An algorithm that at any time $t$  returns  an $(\alpha+\epsilon)$-approximate solution 
to fair center for the current window $W_t$, requiring
working space $m=\BO{k^2 \log \Delta (c/\epsilon)^{D_{W_t}}}$, for a given constant $c$. The algorithm's update (resp., query) time to handle a new arrival (resp., to compute a solution) is $\BO{m}$ (resp., $\BO{km}$), thus \emph{both independent of the window size $|W_t|$}.
\item 
\sloppy
A modification of the above algorithm that
at any time $t$ is able to return  a $\BO{1}$-approximate solution 
to fair center for $W_t$, requiring
working space $m=\BO{k^2 \log \Delta/\epsilon}$, update time $\BO{m}$ and query time $\BO{km}$. Thus, in this algorithm \emph{the 
exponential dependency on the doubling dimension $D_W$ has been 
removed}, at the expense of a weaker, but still constant, approximation ratio. 
\item 
Extensive experimental evidence on real datasets that the above algorithms indeed provide solutions of quality comparable to the best sequential algorithms run on the entire window, \emph{using only a fraction of the space and being orders of magnitude faster}.
\end{itemize}

To the best of our knowledge, ours are the first accurate, space and time efficient sliding window algorithms
for fair center clustering. 
Our algorithms build upon the coreset-based strategy used in \citep{DBLP:conf/icalp/Cohen-AddadSS16,PellizzoniPP22} for the unconstrained $k$-center problem, but 
introduce non trivial, crucial modifications in the coreset construction
to ensure that accurate fair solutions can be extracted from the coreset.  

We remark that although the performance of our most accurate algorithm is expressed in terms of the doubling dimension $D_{W_t}$, the algorithm does not require the knowledge of this parameter to run, which is very desirable in practice, since the doubling dimension is hard to estimate. 
Similarly, the aspect ratio $\Delta$ of the current window needs not be
provided as an explicit input to our algorithms, but it can inferred by maintaining good estimates of the minimum and maximum pairwise distances of the points in $W_t$ (hence, of $\Delta$) without worsening its theoretical and practical performance. 

\paragraph{Organization of the paper} The rest of the paper is structured as follows. Section~\ref{sec:prelim} defines the problem, the computational model and a number of basic notions. Section~\ref{sec-fairkcenter} presents the algorithm (Subsection~\ref{sec-algorithm}) and its analysis (Subsection~\ref{sec-analysis}). Section~\ref{sec-experiments} reports on the experimental results. Section~\ref{sec-conclusions} closes the paper with some concluding remarks.

\section{Preliminaries} \label{sec:prelim}

\paragraph{Problem definition}
Consider a metric space $\X$ equipped with a distance function $d(\cdot,\cdot)$. For $x\in \X$ and $W\subseteq \X$, let $d(x,W) = 
\min\{d(x,w): w\in W\}$ denote the minimum distance of $x$ from a point in $W$. For a given $W\subseteq \X$ and a natural $k$, the classic \emph{$k$-center} clustering problem \citep{Gonzalez85} requires to find a subset $C \subseteq W$ of size at most $k$ minimizing the \emph{radius}
$r_C(W) = \max_{p \in W} d(p, C).$
Observe that any set of centers $C$ defines a natural partition of the points of $W$ into \emph{clusters}, by assigning each point of $W$ to its closest center in $C$ (with ties broken arbitrarily). 

In this paper, we study the following variant of the $k$-center problem, dubbed \emph{fair center}. Assume that each point in $\X$ is associated with one of a finite set of $\ell$ categories (dubbed \emph{colors} in the following). For a given $W\subseteq \X$ and positive integers $k_1, k_2, \ldots, k_{\ell}$, a solution to the fair center problem is  a set  $C\subseteq W$ of centers  minimizing $r_C(W)$, under the additional constraint that $C$ contains at most $k_i$ centers of the $i$-th color, for $1\leq i\leq \ell$. We denote with $\OPT_{W}$ the radius of an optimal fair center solution for $W$.

The above fairness constraint can be envisaged as a special case of the more general class of \emph{matroid constraints}, that have been studied in the context of clustering starting from \citep{ChenLLW16}. 
Given a  ground set $\X$, recall that a matroid $\M$ on $\X$ is a pair $\M = (\X, I)$, where $I \subseteq 2^{\X}$ is a family of \emph{independent sets}  featuring the following two properties:  (a) \emph{downward closure} 
(if $P\in I$ and $P'\subseteq P$ then $P'\in I$); and (b) \emph{augmentation property} (if $P, Q \in I$ and  $|P| > |Q|$ then $\exists x \in P \setminus Q$ such that $Q \cup 
\{x\} \in I$). 
An independent set is \emph{maximal} if it is not a proper subset of any other independent set. Observe that as a consequence of the augmentation property, every maximal independent set in a matroid $\M$ has the same cardinality,
which is denoted with $\rank{M}$. Furthermore, any subset
$W \subseteq \X$ induces a (sub)matroid $(W,I')$, with $I' = \{Y \cap W:  Y \in I\}$. We say that a subset $Y\subseteq W$ is a maximal independent set w.r.t. $W$, if $Y\in I'$ is a maximal independent set of this submatroid. 

Given a matroid $\M = (\X, I)$ and a set $W\subseteq \X$, the \emph{matroid center} problem seeks an \emph{independent set}  $C \in I$ minimizing $r_C(W)$. 
For fixed $k_1, k_2, \ldots, k_\ell$ and $W\subseteq \X$, the constraint to be imposed on the solution of fair center can be seen as a matroid constraint with respect to the so-called \emph{partition matroid} of rank $k=\sum_{i=1}^{\ell}k_i$, where the family of independent sets contains all subsets of $\X$ with at most $k_i$ points of each color $i$, for $1\leq i\leq \ell$. This implies that any algorithm for the  matroid center problem can be immediately specialized to solve the fair center problem. 

\paragraph{Doubling metric spaces.}
Given $W \subseteq \mathcal{X}$,
a point $x \in W$ and a real value $r > 0$, the \emph{ball} of radius $r$
centered in $x$, denoted by $B(x, r) \subseteq W$, is the set $\{p \in W : d(x, p) \le r\}$.
The \emph{doubling dimension} of $W$ is the minimum value $D$ such that,
for all $x \in W$,  $B(x, r)$ is contained in the union of
at most $2^D$ balls of radius $r/2$.
The concept of doubling dimension generalizes the notion of dimensionality of Euclidean spaces and has been used in a wide variety of applications (see~\citep{GottliebKK14,PellizzoniPP22} and references therein).

\paragraph{Sliding windows model}
A \emph{stream} $S$ is a potentially infinite ordered sequence of points from some (metric) space $\mathcal{X}$. 
At each (discrete) time step $t=1,2, \ldots$, a new point $p$ arrives, and we denote with $t(p)$ its arrival time.
Given an integer $n$ and a time $t$, the \emph{window} $W_t \subseteq \X$ at time $t$ of size $n$  is the (multi-)set of the last $n$ points of the stream $S$.
Solving a problem on the sliding windows model entails maintaining data structures that can be queried at any time $t>0$ to return the solution of the problem for the instance $W_t$.  In the sliding window  model, the key performance metrics are
(a) the amount of \emph{space} used to store the data structures;
(b) the \emph{update time} required to handle the arrival of a point $p$ at time $t$; and 
(c) the \emph{query time}, that is the time to extract the solution for window $W_t$ from the data structures.

\section{Fair-center for sliding windows} \label{sec-fairkcenter}
Consider a stream $S$ of colored points from a metric space $\X$ with distance function $d(\cdot,\cdot)$, and let $n>0$ be the target window size. 
We let $\M = (\X,I)$ be a partition matroid defined on $\X$,
whose independent sets are subsets containing at most $\leq k_i$
points of each color $i$, for all $1 \leq i \leq \ell$, and whose rank is $k = \sum_{i=1}^{\ell} k_i$. In this section, we present our algorithm that, at any time $t>0$, is able to return an accurate solution to fair center for the current window $W_t$ of size $n$. The algorithm is described in Subsection~\ref{sec-algorithm} and its accuracy, as well as its time and space requirements, are analyzed in 
Subsection~\ref{sec-analysis}. 

\subsection{Algorithm} \label{sec-algorithm}
The algorithm builds upon the one presented in \citep{PellizzoniPP22} for the unconstrained $k$-center problem in sliding windows, but it introduces crucial modifications which allow to handle the fairness constraint. Also,  a significant simplification of the employed data structures affords more elegant and intuitive correctness and performance analyses with respect to all previous approaches for unconstrained $k$-center \citep{DBLP:conf/icalp/Cohen-AddadSS16,PellizzoniPP22,PellizzoniPP22a}. In this section we provide a full description of the algorithm and its analysis,  pointing out the specific differences with the aforementioned precursor.

At any time step $t$, for a point $p \in S$ with $t(p) \leq t$, its \emph{Time-To-Live} (TTL), denoted as $\TTL(p)$, is the number of remaining steps $\geq t$, in which $p$ will be part of the current window, namely $\TTL(p)= \max\{0,n-(t-t(p))\}$, where $n$ is the size of the window. We say that $p$ is \emph{active} at any time when $\TTL(p) >0$. 
$\TTL(p)$ decreases at every time step, and we say that $p$ \emph{expires} at the time when $\TTL(p)$ becomes 0 (i.e., at time $t(p)+n$).
Let $\dmin$ and $\dmax$ be, respectively, the minimum and maximum pairwise distance between points of $S$, and define $\Delta = \dmax/\dmin$, which we refer to as the \emph{aspect ratio} of $S$. 
For a fixed parameter $\beta>0$, we define the following set of guesses for the optimal radius (which, being a distance between two points,  clearly falls in $[\dmin,\dmax]$):
\[
\Gamma = \left\{
(1+\beta)^i : 
\lfloor \log_{1+\beta} \dmin \rfloor
\leq i \leq
\lceil \log_{1+\beta} \dmax \rceil
\right\}\protect\footnote{For ease of presentation, we assume that $\dmin$ and $\dmax$ are known to the algorithm. However, we remark that the same techniques introduced in \citep{PellizzoniPP22} can be employed
to provide estimates of these quantities and to make $\Gamma$ adaptive to the aspect ratio of the current window, rather than of the entire stream.}.
\]

In broad terms, the algorithm maintains, for each guess $\gamma \in \Gamma$, suitable sets of active points whose overall size is independent of $n$. At any time $t$, from these points it will be possible to identify a guess providing a tight lower bound on the optimal radius, and, based on this guess, to extract a small coreset which embodies a provably good approximate solution to fair center for $W_t$. A query will then compute  such a solution by running the best sequential algorithm available for fair center on the coreset. 
For each $\gamma \in \Gamma$, 
the algorithm maintains two families of active points: 
\emph{validation points} and
\emph{coreset points}.
In turn, each of these two families consists of two (not necessarily disjoint) sets: namely,
$\AVg$ ($v$-\emph{attractors}), and
$\RVg$ ($v$-\emph{representatives}), for validation points; and 
$\Ag$ ($c$-\emph{attractors}), and
$\Rg$ ($c$-\emph{representatives}), for coreset points. These sets are updated after the arrival of each new point of the stream. 
Set $\AVg$ contains at most $k+1$ points at pairwise distance
$\geq 2 \gamma$.
Each $v$-attractor $v\in\AVg$ is paired with a single $v$-representative
denoted as $\repVg(v)$, 
which is %
a recent point of the stream at distance at most $2\gamma$ from $v$, and it is included in $\RVg$.
We remark that there may be $v$-representatives that are not paired with any $v$-attractor. Indeed, when a $v$-attractor $v\in\AVg$ 
expires, its $v$-representative $\repVg(v)$ is not further updated, but it remains in
$\RVg$ until it expires or it is expunged by a suitable clean-up procedure. 

A key difference with the algorithm in \citep{PellizzoniPP22} is in the choice of coreset points, which must now be empowered to account for the fairness constraint. 
Let $\delta \in (0,1)$ be a fixed, user-defined precision parameter that is closely related to the resulting approximation quality of the algorithm (see Theorem \ref{thm:approx}).
The set $\Ag$ contains points at pairwise distance
$\geq \delta \gamma/2$, which is a factor $\delta/4$ less than the least distance between $v$-attractors in $\AVg$. 
There is no fixed upper bound
on its size, which, however, will be conveniently bounded by the analysis. Upon arrival, a new point $p$ is either \emph{attracted} by a conveniently chosen point in $\Ag$ at distance $\leq \delta \gamma/2$, if any, or is added to $\Ag$ (and attracts itself). For each $a \in \Ag$, 
the algorithm maintains a \emph{set} of $c$-representatives $\repsCg(a)$,
which is a maximal independent set of active points attracted by $a$ with the longest remaining lifespan, and is included in $\Rg$.
For $i\in [1, \ell]$, we let
$\repsCg^i(v) \subseteq \repsCg(v)$ denote the 
(at most $k_i$) points of $\repsCg(v)$ of  color $i$.
After a $c$-attractor $a\in\Ag$ 
expires, its set of $c$-representatives $\repsCg(a)$ 
is not further updated, but each of its elements remains in
$\Rg$ until it expires or it is expunged by the aforementioned clean-up procedure.  

Clearly, all of the above sets of validation and coreset points evolve with time and, whenever we need to refer to one such set at a specific time $t$, we will add $t$ as a second subscript. We remark that in \citep{PellizzoniPP22}, representatives of expired $v$- and $c$-attractors were represented as different sets of point called "orphans". Our novel approach avoids this distinction, which yields a substantial simplification of the algorithm and of its analysis. 

An informal intuition behind the use of validation and coreset points is the following. Recall that $k=\sum_{i=1}^\ell k_i$.
We say that a guess $\gamma \in \Gamma$ is \emph{valid} at
time $t$, if $|\AVgt[t]| \leq k$. It is easy to see that if $\gamma \in \Gamma$ is not valid, (i.e., $|\AVgt[t]| \geq k+1$), then $\gamma$ is a lower bound to the optimal radius for unconstrained $k$-center on $W_t$, hence to the optimal radius for fair center on $W_t$, which can only be as large. 
The main purpose of validation points is to keep track of current valid guesses. Instead, coreset points provide, for each valid guess, a reasonably sized subset of the window points, so that any other window point is “well represented”  by a nearby coreset point of the same color, which ensures that an accurate fair solution on the coreset points is also an accurate solution for the entire window. 
Thus, running a sequential algorithm solely on the coreset points of a valid guess will yield a close approximation of the solution obtainable by running the algorithm on the entire window, while significantly reducing time and memory usage.
For each guess, the number of coreset points will depend on  $k$,  on the precision parameter $\delta$, and, most importantly,  on the doubling dimension $D_{W_t}$ of the current window. However, we remark that while $k$ and $\delta$ are 
parameters that must be given in input, 
the algorithm does not require the knowledge of the doubling dimension $D_{W_t}$, which will be used only in the analysis.

In what follows, we describe in more details the operations performed at time step $t$ to handle the arrival of a new point $p$ and, if required, to compute a solution to fair center for the current window $W_t$. 

\paragraph{Update procedure} 
\begin{algorithm}[h]
\caption{\textsc{Update}($p$)}
\label{alg:update}
\SetAlgoVlined
\LinesNumbered

\Let{$i$}{color of $p$}

\Let{$x$}{point expired when $p$ arrives}

\ForEach{$\gamma \in \Gamma$}{
    Remove $x$ from any set $AV_\gamma,A_\gamma,RV_\gamma,R_\gamma$ containing it\;

    \tcc{Identify which attractors $p$ can be a representative for}
    \Let{$EV$}{$\{ v \in AV_\gamma : d(p, v) \le 2 \gamma \}$}\;
    \Let{$E$}{$\{ v \in A_\gamma : d(p, v) \le \delta \gamma / 2 \}$}\;
       \tcc{Assign $p$ to a $v$-attractor $\psi_\gamma(p)$}

    \If{$EV = \emptyset$}{
        \Let{$AV_\gamma$}{$AV_\gamma \cup \{p\}$} \;
    \Let{$\psi_\gamma(p)$}{$p$};
    \Let{$\repVg(\psi_\gamma(p))$}{$p$}\;
    \Let{$RV_\gamma$}{$RV_\gamma \cup \{\repVg(\psi_\gamma(p))\}$}\;
        \textsc{Cleanup}($p, \gamma$)\;
    }
    \Else {
          \Let{$\psi_\gamma(p)$}{arbitrary element of $EV$}\;
         \Let{$repV_\gamma(\psi_\gamma(p))$}{$p$}
    }
     \tcc*[h]{Assign $p$ to a $c$-attractor $\phi_{\gamma}(p)$}
     
    \If{$E = \emptyset$}{\label{ln:coreset-update}
        \Let{$A_\gamma$}{$A_\gamma \cup \{p\}$}\; 
    \Let{$\phi_\gamma(p)$}{$p$};
    \Let{$\repsCg(\phi_\gamma(p))$}{$\{p\}$}\; 
    \Let{$\Rg$}{$\Rg \cup \{\repsCg(\phi_\gamma(p))\}$}\;
    } \Else {
       
        \Let{$\phi_{\gamma}(p)$}{$\argmin_{a \in E} |\repsCg^i(a)|$}\; \label{ln:c-attr}

        \tcc*[h]{Update $\repsCg^i(\phi_{\gamma}(p))$}

        \Let{$\repsCg^i(\phi_{\gamma}(p))$}{$\repsCg^i(\phi_{\gamma}(p)) \cup \{p\}$}\;
        \label{ln:coreset-update}
        \If{$|\repsCg^i(\phi_{\gamma}(p))| > k_i$}{
            \Let{$o_{rem}$}{$\argmin_{o \in \repsCg^i(\phi_{\gamma}(p))} TTL(o)$}\; 
            Remove $o_{rem}$ from $\repsCg(\phi_{\gamma}(p))$ in $\Rg$\;
            \label{ln:cat-overflow-end}
        }
    }
}
\end{algorithm}

\begin{algorithm}[h]
\caption{\textsc{Cleanup}($p, \gamma$)}
\label{alg:cleanup}
\SetAlgoVlined
\LinesNumbered
\If(\tcc*[h]{Remove the oldest $v$-attractor from $AV_\gamma$}){$|AV_\gamma| = k + 2$} {
    \Let{$v_{old}$}{$argmin_{v \in AV_\gamma} TTL(v)$}\;
    \Let{$AV_\gamma$}{$AV_\gamma \setminus \{v_{old}\}$}\;
}
\If(\tcc*[h]{remove unneeded points}){$|\AVg| = k+1$}{
    \Let{$t_{min}$}{$\min_{v \in \AVg} TTL(v)$}\;
    Remove all $q$ with $TTL(q) < t_{min}$ from $\Ag$, $\RVg$, and $\Rg$  
    \label{ln:expunge} \;
}

\end{algorithm}
The arrival of a new point $p$ is handled by procedure \textsc{Update}($p$) (Algorithm~\ref{alg:update}). For every guess $\gamma \in \Gamma$, the data structures are updated as follows. If $p$ is at distance
greater than $2\gamma$ from any other $v$-attractor, then $p$ is added to  $\AVg$, and to $\RVg$ as representative of itself,
otherwise it becomes the new representative of %
some (arbitrary) $v$-attractor %
$v$ with $d(v,p) \leq 2 \gamma$. In the former case, a procedure \textsc{Cleanup}($p,\gamma$) (Algorithm~\ref{alg:cleanup}) is invoked to reduce the size of the various sets as follows. If $|\AVg| = k+2$ the $v$-attractor with minimum TTL is removed from $\AVg$. 
After this, if  $|\AVg| = k+1$ all  
$c$-attractors, $v$-representatives, and $c$-representatives with TTL less than the minimum TTL $t_{min}$ of a
$v$-attractor are removed from $\Ag$.  These removals are justified by the fact that if $|\AVg| = k+1$, then
$\AVg$ acts as a certificate that $\gamma$ is not a valid guess until time $t+t_{min}$, hence, for that guess,  there is no need to keep points that expire earlier than that. 
For every point $p$, we denote with $\psi_{\gamma}(p)$ the $v$-attractor $v$ such that $p = \repVg(v)$ after the execution of \textsc{Update}$(p)$. in fact, $\psi_{\gamma}(p)$ is not explicitly stored in the data structures, but is solely needed for the analysis.
Coreset points are updated as follows. 
If $p$ is at distance
greater than $\delta \gamma/2$ from any other $c$-attractor, then $p$ is added to $\Ag$.
Otherwise, suppose that $p$ is of color $i$. Then $p$ is attracted by $a$ and added to the set $\repsCg(a)$, where
$a$ is chosen as the $c$-attractor at distance
at most $\delta \gamma/2$ from $p$ with minimum
$|\repsCg^i(a)|$. In case $|\repsCg^i(a)| > k_i$  (i.e., $|\repsCg^i(a)| = k_i+1$), 
the representative in $|\repsCg^i(a)|$ with minimum TTL is removed from the set. Observe that, unlike $AV_\gamma$, the size of $\Ag$ can grow  larger than $k+2$. The maximum size of these sets, however, will be conveniently upper bounded by the analysis. 
For every point $p$, we denote with $\phi_{\gamma}(p)$ the $c$-attractor such that $p \in \repsCg(\phi_{\gamma}(p))$, after the execution of \textsc{Update}$(p)$. Again,  
 $\phi_{\gamma}(p)$ is not explicitly stored in the data structures, but is solely needed for the analysis.

\paragraph{Query procedure}

\begin{algorithm}
\caption{\textsc{Query}()}
\label{alg:query}
\SetAlgoVlined
\LinesNumbered

\For{increasing values of $\gamma \in \Gamma$ such that $|\AVg| \le k$}{
    \Let{$C$}{$\emptyset$}\;
    \ForEach{$q\in \RVg$}{
        \lIf{$C = \emptyset \vee d(q, C) > 2\gamma$}{
            \Let{$C$}{$C \cup \{q\}$}
        }
        \lIf{$|C| > k$}{Break and move to the next guess
        }
    }
    \If{$|C| \le k$}{
    \label{ln:solution}
        \Return $\mathcal{A}(\Rg)$
        \tcc*[]{$\mathcal{A} = $ sequential fair center algorithm}
    }
}
\end{algorithm}

At any time $t$, to obtain a  solution to the fair center problem for $W_t$,
procedure \textsc{Query}() is invoked (Algorithm~\ref{alg:query}).
First, a valid guess $\gamma$ is identified 
such that a $k$-center clustering of radius $\leq 2 \gamma$ for the points in 
$\RVg$ is found, and  for any smaller guess $\gamma' \in \Gamma$ with $\gamma' < \gamma$, there are $k+1$ points in either $AV_{\gamma'}$ or $RV_{\gamma'}$ with pairwise distance $> 2 \gamma'$.
As it will be shown in the analysis, this ensures that the coreset points in $\Rg$ contain a good solution to fair center
for $W_t$, which is then computed by invoking the best sequential fair center algorithm  $\mathcal{A}$.

\subsection{Analysis} \label{sec-analysis}
The following lemma shows that at any time $t$,
the points  $\RVg$ (resp., $\Rg$)
are within distance $4 \gamma$
(resp., $\delta \gamma$) from all points of the entire $W_t$, when
$\gamma$ is valid, or of a suitable suffix
of $W_t$, otherwise. (The lemma is similar to \citep[Lemma 1]{PellizzoniPP22} but the proof given below is completely new and applies to the simpler version of the update procedure devised in this paper.)  

\begin{lemma} \label{lem:invariants}
    For every $\gamma \in \Gamma$ and $t>0$, the following
    properties hold after the execution of the \textsc{Update} procedure for the point arrived at time $t$:
 \begin{enumerate}
  \item\label{lemprop1} 
  If $|\AVgt[t]| \le k$, then $\forall q\in W_t$ we have:
    \begin{itemize}
        \item[]{\rm \bf (a)}
             $d(q, \RVgt[t]) \le 4\gamma$;
        \item[]{\rm \bf (b)}
        \label{case1b}
              $d(q, \Rgt[t]) \le \delta\gamma$;
    \end{itemize}
 \item\label{lemprop2} If $|\AVgt[t]| > k$, then $\forall q\in W_t$ such that $t(q) \ge \min_{v\in \AVgt[t]} t(v)$ we have:
    \begin{itemize}
        \item[]{\rm \bf (a)}
            $d(q, \RVgt[t]) \le 4\gamma$
        \item[]{\rm \bf (b)}
            $d(q, \Rgt[t]) \le \delta\gamma$;
    \end{itemize}
    \end{enumerate}
\end{lemma}

\begin{proof}
For every $\gamma \in \Gamma$ and $t>0$, we say that a point $q \in W_t$ is
\emph{relevant for $\gamma$ and $t$} if $|\AVgt[t]| \leq k$
or $|\AVgt[t]| > k$ and $t(q) \ge \min_{v\in \AVgt[t]} t(v)$.
Clearly, the proof concerns only relevant points.  Also, recall that for every point $p$ of the stream, in \textsc{Update}$(p)$ we set $\psi_{\gamma}(p)$ (resp., $\phi_{\gamma}(p)$) to its assigned $v$-attractor (resp., $c$-attractor). Clearly, $d(p,\psi_{\gamma}(p)) \leq 2 \gamma$,  and $d(p,\phi_{\gamma}(p)) \leq \delta \gamma/2$.  (In fact, the values $\psi_{\gamma}(p)$ and $\phi_{\gamma}(p)$ are not explicitly stored in the data structures, as they are just needed for the analysis.) 

Let us focus on an arbitrary guess $\gamma$. We will  argue  that, at any time $t > 0$, for every point $q$ which is relevant for $\gamma$ and $t$, there exists a point $x \in \RVgt[t]$ and a point $y \in \Rgt[t]$
such that $\psi_{\gamma}(x) = \psi_{\gamma}(q)$ and $\phi_{\gamma}(y) = \phi_{\gamma}(q)$. The lemma will follow immediately since
$\psi_{\gamma}(x) = \psi_{\gamma}(q)$ implies
$d(q,x) \leq 4 \gamma$, and $\phi_{\gamma}(y) = \phi_{\gamma}(q)$ implies $d(q,y) \leq \delta \gamma$. 

Let us first concentrate on validation points. %
If $q$ is the point arrived at time $t$ (i.e., $t(q)=t$) then $q \in \RVgt[t]$, and it suffices to choose $x=q$. Otherwise, it is immediate to argue that since $q$ is relevant at time $t$, it has also been relevant at all times $t'$, with $t(q)\leq t' <t$, which implies that $q$ has never been been removed from $\RVgt[t']$ during some invocation of \textsc{Cleanup}, which only removes non-relevant points. Consequently, since  $q$ entered $\RVgt[t(q)]$, either $q$ is still in  $\RVgt[t]$ (and we set $x=q$), or $q$ was eliminated from $\RVgt$  due to the arrival of some point $x_1$ at time $t(x_1)>t(q)$.  Observe that $x_1$ is also relevant. If $x_1\in \RVgt[t]$, we set $x=x_1$, otherwise, we apply the same reasoning to $x_1$. Iterating the argument,  we determine a sequence of relevant points $x_1, x_2, \ldots, x_r$, with $t(q)<t(x_1)<\ldots <t(x_r)$ and 
$\psi(q)=\psi(x_1) = \ldots = \psi(x_r)$, with 
$x_r\in \RVgt[t]$, and we set $x=x_r$. 

The argument to show that there exists $y \in \Rgt[t]$ such that 
$\phi_{\gamma}(y) = \phi_{\gamma}(q)$ is virtually identical, and is just sketched for completeness. As argued above, since $q$ is relevant for $\gamma$ and $t$, it was never eliminated during some invocation of \textsc{Cleanup}. Thus, since  $q\in \Rgt[t(q)]$, either $q\in \Rgt[t]$ (and we set $y=q$) or we can identify a sequence of relevant points $y_1, y_2, \ldots, y_r$ of the same color as $q$,  with $t(q)<t(y_1)<\ldots <t(y_s)$ and 
$\phi(q)=\phi(y_1) = \ldots = \phi(y_s)$, with 
$y_s\in \Rgt[t]$, and we set $y=y_s$. 
\end{proof}
The next lemma (proof omitted) is a straightforward adaptation
of \cite[Lemma 3]{DBLP:journals/jbd/CeccarelloPPS23},
and provides an important technical tool,
which will be needed to show that the coreset extracted in {\sc Query}() contains a good solution to fair   center for the current window. 

\begin{lemma}
\label{lem:coreset-conditions}
Let $W_t$ be the window of points at time step $t$, and
let $I_{W_t}$ be the family of independent sets of the
partition matroid defined on $W_t$ (i.e., the feasible solutions to fair center for $W_t$). 
Let $Q \subseteq W_t$ ba a coreset that satisfies the following conditions:
\begin{itemize}
    \item[]\hspace*{-0.3cm}{\rm \bf (C1)} $d(p, Q) \le \delta\gamma, \quad \forall p \in W_t$
    \item[]\hspace*{-0.3cm}{\rm \bf (C2)} For each independent set $X\in I_{W_t}$ there exists an injective mapping 
        $\pi_X: X \to Q$ such that:
        \begin{itemize}
        \item $\{\pi_X(x) : x \in X\} \subseteq Q$ is an independent set w.r.t. $Q$
        \item for each $x \in X$, $d(x, \pi_X(x)) \le \delta\gamma$
        \end{itemize}
\end{itemize}
Then:
\begin{itemize}
\item[]\hspace*{-0.3cm}{\rm \bf (P1)} There exists a solution $S\subseteq Q$ to fair center for $W_t$ 
of radius $\leq \OPT_{W_t} + 2\delta\gamma$
\item[]\hspace*{-0.3cm}{\rm \bf (P2)} Every solution to fair center for $Q$
of radius $r$  is also a solution of cost at most $r + \delta\gamma$ for $W_t$.
\end{itemize}
\end{lemma}
We are now ready to present the main novel ingredient of the analysis of our algorithm. Recall that for any point $p$ of the stream and every guess $\gamma$, {\sc Update}$(p)$ assigns $p$ to a
$c$-attractor $\phi_{\gamma}(p) \in \Ag$,
adding $p$ to $\repsCgt[t](\phi_{\gamma}(p))$. Note that $\phi_{\gamma}(p)=p$ if $p$ was added to $\Ag$. Also, for any point $a$ which at time $t(a)$ was added to $\Ag$, and for any $t \geq t(a)$, we define
\[
W_{\gamma,t}(a) = \{x \in W_t \; : \; \phi_{\gamma}(x) = a \}.
\]
Observe that from time $t(a)$ until the time when $a$ is expunged from $\Ag$, $W_{\gamma,t}(a)$ grows with the arrival of each new point $c$-attracted by $a$. After $a$ leaves $\Ag$, 
$W_{\gamma,t}(a)$ progressively looses points as they expire.  Also, for $t\geq t(a)$,  the set $\repsCgt[t](a)$ is updated at each arrival of a new point attracted by $a$, until $a$ is expunged from $\Ag$. After that happens,
the set $\repsCgt[t](a)$ is not further updated with new points, while its elements will then be gradually expunged from $\Rg$ as they expire. We have:
\begin{lemma}\label{lem:misrep}
For any time $t$ and valid guess $\gamma \in \Gamma$, the following holds. 
Let $a \in S$ be a point which arrived at $t(a) \leq t$ and, upon arrival, was added to $\Ag$. Then $W_{\gamma,t}(a) \cap R_{\gamma,t})$ is a maximal independent set w.r.t.\  $W_{\gamma,t}(a)$.
\end{lemma}
\begin{proof}
It is immediate to see that at each time $t\geq t(a)$ for which $a \in \Agt[t]$, 
the set $\repsCgt[t](a) = W_{\gamma,t}(a) \cap R_{\gamma,t})$
is always a maximal independent set w.r.t.\ $W_{\gamma,t}(a)$,  since a point  $x\in W_{\gamma,t}(a)$ is always inserted into  $repsC_{\gamma, t(x)}(a)$, and may cause the expunction from the set of a single  point  of the same color $i_x$ (the oldest such point), only when there are already $k_{i_x}$ points of color $i_x$ present in $repsC_{\gamma, t(x)-1}(a)$. 
Let $\hat{t} > t(a)$ be the time when $a$ is expunged from 
$A_{\gamma,\hat{t}}$. As already observed, for every $t\geq \hat{t}$, $W_{\gamma,t}(a)$ does not acquire new points and shrinks due to the natural expiration of the points $c$-attracted by $a$. If points in $\Rgt[t]$ were removed only on their expiration, then the maximality of 
$W_{\gamma,t}(a) \cap \Rgt[t]$ w.r.t. $W_{\gamma,t}(a)$ would immediately follow by the fact that, as long as $a$ is
in $\Agt[t]$, 
for each color $i$, $\repsCgt[t](a)$ maintains the most recently arrived points of that color. However, whenever $\gamma$ becomes  an invalid guess, extra points from $\Rgt[t]$ may be expunged (Line~\ref{ln:expunge} of procedure {\sc Cleanup}). 
However, these expunged  points, being older than the oldest point in $AV_{\gamma,t}$, will all have expired by the time step $t'> t$ (if any) when $|AV_{\gamma,t'}|$ drops below $k+1$, and thus $\gamma$ becomes again a valid guess, so their premature elimination will not affect the maximality of $W_{\gamma,t'}\cap \Rgt[t']$.
\end{proof}

The following theorem establishes the quality of the solution returned by Procedure {\sc Query}.
\begin{theorem} 
\label{thm:approx}
Let $\alpha$ be the approximation ratio featured by the sequential fair center algorithm $\mathcal{A}$ used in Procedure {\sc Query}().
For fixed $\epsilon \in (0,1)$ and $\beta > 0$, by setting
\[
\delta = {\epsilon \over (1+\beta)(1+2\alpha)}
\]
we have that if Procedure {\sc Query} is run at time $t$,
then the returned solution is an $(\alpha + \epsilon)$-approximate solution to fair center for the current window $W_t$.
\end{theorem}
\begin{proof}
Let $\hat{\gamma}$ be the guess such that the solution returned
by \textsc{Query} is computed by running
$\mathcal{A}$ on the coreset $Q=\Rg$. We first show that  conditions (C1) and (C2) of Lemma~\ref{lem:coreset-conditions} hold for $Q$. First observe that by construction $\hat{\gamma}$ is chosen so that $|AV_{\hat{\gamma},t}| \leq k$, therefore condition (C1) holds by virtue of Property~\ref{lemprop1}.(b) of Lemma~\ref{lem:invariants}.
To show that (C2) also holds, consider any independent set $X$ w.r.t. $W_t$. We construct the required injective mapping 
 $\pi_X: X \to Q$ incrementally, one point at a time. 
Let $X=\{x_j \; : \; 1 \leq j \leq |X|\}$. Suppose that we
have fixed the mapping for the first $h \geq 0$ elements of $X$, and
assume, inductively, that 
$$
Y(h)= \{\pi_X(x_j) \: : \; 1 \leq j \leq h\} \cup \{x_j \; : \; h < j \leq |X|\}
$$
is an independent set w.r.t. $W_t$ of size $|X|$, and such
that, for $1 \leq j \leq h$, $\pi_X(x_j) \in Q$ and 
$d(x_j,\pi_X(x_j)) \leq \delta \gamma$.
We now show how to extend the mapping to index $h+1$. We distinguish between two cases. If $x_{h+1} \in Q$, then we simply set $\pi_X(x_{h+1}) = x_{h+1}$. Since  $Y(h+1)=Y(h)$ and $d(x_{h+1},\pi_X(x_{h+1})) = 0 \leq \delta \gamma$, the mapping is correctly extended to index $h+1$. Conversely, if $x_{h+1} \notin Q$, observe that $x_{h+1}\in W_{\hat{\gamma},t}(a)$, where  $a=\phi(x_{h+1})$. Since $\hat{\gamma}$ is a valid guess, by Lemma~\ref{lem:misrep} it follows that 
$Z= W_{\hat{\gamma},t}(a) \cap Q$ is a maximal independent set w.r.t.\ $W_{\hat{\gamma},t}(a)$. Then, it will always be possible to set $\pi_X(x_{h+1})=y_{h+1}$, where $y_{h+1}$ is a point of $Z$ of the same color as $x_{h+1}$, and such that $y_{h+1}\neq \pi_X(x_j)$, for $1\leq j\leq h$, or otherwise the number of points in $X\cap W_{\hat{\gamma},t}(a)$ of the same color as $x_{h+1}$ would have to be larger than those in $Z$, contradicting the maximality of the latter subset w.r.t.\ $W_{\hat{\gamma},t}(a)$. By the properties of the partition matroid, 
it follows that $Y(h+1)$ is still an independent set w.r.t.\ $W_t$. 
Also, since both $x_{h+1}$ and $\pi_X(x_{h+1})=y_{h+1}$ belong to $W_{\hat{\gamma}, t}(a)$, we have that $d(x_{h+1},\pi_X(x_{h+1})) \leq d(x_{h+1},a)+d(a,\pi_X(x_{h+1}) \leq \delta\gamma$. Therefore, given that, by construction, 
$\pi_X(x_{h+1}) \in Q$, the mapping is correctly extended to index $h+1$. By iterating the argument up to index $|X|$, 
we have that $Y(|X|)\subseteq Q$, and
 we conclude that condition (C2) holds for $Q\subseteq W_t$. Therefore, by
Lemma~\ref{lem:coreset-conditions}, since (C1) and (C2) hold for $Q$, Properties (P1) and (P2) also hold. Property (P1) implies that the optimal solution to fair center for $Q$ has radius $\leq \OPT_{W_t}+2 \delta \hat{\gamma}$. Therefore, 
when invoked on $Q$, algorithm $\mathcal{A}$ will return a solution of radius $\leq \alpha (\OPT_{W_t}+2 \delta \hat{\gamma})$ for $Q$, which,
by Property (P2), is also a solution of radius
$\leq \alpha \OPT_{W_t}+(1+2 \alpha) \delta \hat{\gamma}$ for $W_t$. Now, let $r^*_k$ be the radius of the optimal solution to unconstrained $k$-center for $W_t$. 
It is easy to see that $\hat{\gamma} \leq (1+\beta) r^*_k$. Indeed, the inequality is trivial if $\hat{\gamma} = \dmin$, while, otherwise, it follows from the fact that Procedure \textsc{Query} discarded guess $\gamma = \hat{\gamma}/(1+\beta)$ because $k+1$ points of $W_t$ at pairwise distance $2 \gamma$ exist and, clearly, two of these points must be closest to the same center of the optimal solution of 
unconstrained $k$-center for $W_t$ of radius $r^*_k$. Since 
$r^*_k \leq \OPT_{W_t}$, we have that 
$\hat{\gamma} \leq (1+\beta) \OPT_{W_t}$. 
The above discussion and the choice of $\delta$ immediately imply that \textsc{Query}
will return a solution whose radius, with respect to the entire window $W_t$, is at most 
\[
\alpha \OPT_{W_t}+(1+2 \alpha) \delta \hat{\gamma} \leq 
\alpha \OPT_{W_t} + (1+2 \alpha) {\epsilon \over (1+\beta)(1+2\alpha)} (1+\beta) \OPT_{W_t}
= (\alpha + \epsilon) \OPT_{W_t}.
\]
\end{proof}

By using the algorithm of \citep{DBLP:conf/icml/JonesNN20} as algorithm $\cal A$ in {\sc Query}, we obtain:
\begin{corollary}
For fixed $\epsilon \in (0,1)$, at any time $t>0$, procedure \textsc{Query} can be used to compute a $(3+\epsilon)$-approximate solution to fair center for window $W_t$.
\end{corollary}

The following theorem bounds the size of the working memory required by our algorithm.
\begin{theorem}\label{thm:space}
Under the same parameter configuration of Theorem~\ref{thm:approx},
at any time $t$ during the processing of stream $S$, the sets stored in in the working memory contain
\[
\BO{k^2 \frac{\log \Delta}{\log(1+\beta)}\left(\frac{c}{\epsilon}\right)^{D_{W_t}}}
\]
points overall, where $c=32 (1+\beta)(1+2\alpha)$, $D_{W_t}$ is the doubling dimension of the current window $W_t$ and $\Delta$ is the aspect ratio of $S$.
\end{theorem}
\begin{proof}
Consider a time $t$ and a guess $\gamma \in \Gamma$. 
The following facts can be proved through the same arguments
used in the proof of
\cite[Theorem 2]{PellizzoniPP22} for the case of unconstrained $k$-center on a window of doubling dimension $D_{W_t}$.
\begin{fact} \label{fact1}
At each time $t$, $|\AVgt|\leq k+1$ and $\RVgt\leq 2(k+1)$. 
\end{fact} 
\begin{fact} \label{fact2}
$\Ag$ contains at most $k' = 2(k+1)(32/\delta)^{D_{W_t}}$ points. 
\end{fact} 
It remains to upper bound the size of $\Rg$. Clearly, by Fact~\ref{fact2}, there can be altogether at most  $k \cdot k'$ points in $\Rg$ contained in the representative sets $\repsCg$ of points currently in $A_\gamma$. We are then left to bound the number of points contained in representative sets of points not in $A_\gamma$ at the current time. Call $O_\gamma$ the union of such sets.  Let $a_i$, with $i=1,2, \ldots$ be an enumeration of all points that upon arrival have been added to $\Ag$. For every $i \geq 1$
it must hold that $a_i$ has expired or has been expunged from $\Ag$
by the time $a_{i+k'+1}$ enters $\Ag$, or otherwise $\Ag$ would have size greater than $k'$ at time $t(a_{i+k'+1})$, which would 
contradict Fact~\ref{fact2}. Consequently, any point
$x$ added to $\repsCg(a_i)$ must have arrived while $a_i$
was still in $\Ag$, hence before $a_{i+k'+1}$, thus
$\TTL(x) < \TTL(a_{i+k'+1})$ at any time when they are both active. Now, consider the current time, and let $a_j$ be the point which was most recently removed from $\Ag$. By the above property, any point $x$ 
that belonged to $\repsCg(a_{\ell})$, with $\ell \leq j-(k'+1)$,
arrived prior to $a_j$, hence cannot be in memory at time $t$.
Hence, $\Og$ can only comprise points that belonged
to $\repsCg(a_{\ell})$, with $j-(k'+1) < \ell \leq j$, which immediately implies that $|\Og| \leq k \cdot k'$.
The theorem follows by considering that  in Theorem~\ref{thm:approx} we fixed $\delta = \epsilon/((1+\beta)(1+ 2\alpha))$, and that the number of guesses $|\Gamma|= \BO{\log \Delta/\log (1+\beta)}$. 
\end{proof}

\sloppy
For what concerns  running time, it can be easily argued that Procedure \textsc{Update} requires time linear in the aggregate size of the stored sets, while the Procedure \textsc{Query} first executes 
$\BO{\log \Delta/\log (1+\beta)}$ attempts at discovering suitable unconstrained $k$-clusterings on sets of $O(k)$ points, each requiring $O(k^2)$ time, and then invokes the sequential fair center algorithm $\cal A$ on some $\Rg\cup\Og$, for some guess $\gamma\in \Gamma$. 
The above discussion immediately implies the following theorem.
\begin{theorem}\label{thm:time}
Under the same parameter configuration of Theorem~\ref{thm:space}, Procedure \textsc{Update} runs in time
\[
\BO{k^2 \frac{\log \Delta}{\log(1+\beta)}\left(\frac{c}{\epsilon}\right)^{D_{W_t}}}
\]
while procedure \textsc{Query} runs in time
\[
\BO{
k^2 \frac{\log\Delta}{\log(1+\beta)} +
T_{\mathcal{A}}\left(k^2 \left(\frac{c}{\epsilon}\right)^{D_{W_t}}\right)
}
\]
where $T_{\mathcal{A}}(m)$ is the  running time of sequential algorithm $\mathcal{A}$ on an instance of size $m$.
\end{theorem}
It is important to remark that for windows of constant doubling dimension, both space and time requirements of our  algorithm  are independent of the window size $n$. 

The space and time requirements of our algorithm can be made independent of the dimensionality of the stream, at the expense of a worsening in the approximation guarantee. We have:
\begin{corollary}\label{cor:gen}
Procedures {\sc Update}, {\sc Cleanup} and {\sc Query} can be modified so to yield, at any time $t$, a $(31+O(\epsilon))$-approximate solution, storing $\BO{k^2\log\Delta/\epsilon}$ points in the working memory, and requiring $\BO{k^2\log\Delta/\epsilon}$ update  and $\BO{k^2\log\Delta/\epsilon+T_{\cal A}(k^2)}$ query time. 
\end{corollary}
\begin{proof}
 We  recast our algorithm by eliminating the two sets of coreset points and all operations related to their maintenance. Also, rather than a simple point, we make sure to maintain in  $\repVg(v)$ the most recent maximal independent set of points attracted by $v$. Procedures  {\sc Update} and {\sc Cleanup} are modified accordingly. 
 Finally, {\sc Query} computes the final clustering by selecting the right guess $\gamma$ as before and then applying algorithm $\cal A$ to $RV_\gamma$ rather than $R_\gamma$. 
 The bounds on time and space are immediately obtained by noticing that the size of the set  $\RVg$ is no more than a factor $k$ larger than its  size in the original algorithm. The bound on approximation can be obtained by setting $\beta=\epsilon$ and by a straightforward adaptation of 
the analysis of the original algorithm for the case $\delta=4$, a value for which the 
set of $v$-points on which $\cal A$ is invoked coincides with the old set $R_\gamma$.
\end{proof}

\section{Experiments} \label{sec-experiments}

We evaluate experimentally our coreset-based algorithm against state of the art baselines, aiming at answering the following questions:
\begin{itemize}
\item How does the coreset size impact the performance of our algorithm? (Subection~\ref{sec:delta-dependency})
\item How do the algorithms behave on different window lengths? (Subection~\ref{sec:window-size-dependency})
\item How is the performance of our algorithm affected by the actual dimensionality of the dataset? 
(Subection~\ref{sec:dimensionality-dependency})
\end{itemize}
Below, we provide details of the experimental setup, namely, a description of the datasets, the implemented algorithms and baselines, and the performance metrics used in our comparisons.

\paragraph{Datasets} 
We consider three real-world datasets.
\dataset{PHONES}~\footnote{UCI repository: \url{https://doi.org/10.24432/C5689X}} is a dataset of sensor data from smartphones:
each of the 13\,062\,475 points represents the position of a phone in three dimensions, and is labeled with 
one of $\ell=7$ categories corresponding to user actions: stand, sit, walk, bike, stairs up,
stairs down, null. The aspect ratio of the entire dataset is $\approx 6.4\cdot 10^5$;
\dataset{HIGGS}~\footnote{UCI repository: \url{https://doi.org/10.24432/C5V312}} contains 11 million 7-dimensional points representing simulated high-energy particles,  featuring an aspect ratio of $\approx 2.3\cdot 10^4$. 
Each point is labeled as either \emph{signal} or \emph{noise} ($\ell=2$);
\dataset{COVTYPE}~\footnote{UCI repository: \url{https://doi.org/10.24432/C50K5N}} reports cartographic variables, with each of the 581\,012 54-dimensional observations associated with
one out of $\ell=7$ possible forest cover types.
The aspect ratio is $\approx 3.1\cdot 10^3$. 
Besides these  real-world datasets, some experiments are run on suitably crafted synthetic data, which will be specified later.  

\paragraph{Algorithms}
As baseline algorithms we consider the matroid center algorithm by \cite{ChenLLW16} (\algo{ChenEtAl}),
and the fair center algorithm by \cite{DBLP:conf/icml/JonesNN20}
(\algo{Jones}).
Note that both are sequential algorithms
for the problem. In the sliding window setting, upon a query, we run the baseline algorithms on \emph{all} points of the current window. Thus, in principle, these provide the most accurate solutions, at the expense of high space/time complexity.   
Algorithm \algo{Ours} is the implementation of our algorithm with knowledge of the minimum and maximum pairwise distances, $\dmin$ and $\dmax$, between points of the stream, and invokes baseline $\cal{A} = \mbox{\algo{Jones}}$ in \algo{Query} only on the points of the selected coreset. Since the aspect ratio of a stream is often unknown, we also implemented \OursObl, where running estimates of $\dmin$ and $\dmax$ \emph{for the current window} are obtained by means of the techniques of~\citep{PellizzoniPP22}, based on a 
sliding-window diameter-estimation algorithm.
Upon update and query, \OursObl considers only the guesses that are within the
current $[\dmin,\dmax]$ interval.
For both \algo{Ours} and \algo{OursOblivious}, the progression of the guess values
is defined fixing $\beta=2$. In fact, we found that varying this parameter does not significantly influence the results.
The precision parameter $\delta$ will instead varies in the set $\delta \in \{0.5, 1, 1.5, 2, 2.5, 3, 3.5, 4\}$, with the understanding that lower values result in finer-grained (i.e. larger) coresets, while $\delta=4$ is equivalent to using a coreset comparable in size to the validation set (i.e., the one yielding the result of Corollary~\ref{cor:gen}). This set of values for $\delta$ indeed allows spanning a very wide range of coreset sizes.
For the window size, we consider values between 10\,000 and 500\,000 points.
The cardinality constraints on the $\ell$ colors are set so that $\sum_i^\ell k_i = 14$
and that $k_i$ is proportional to the number of points of color $i$ in the entire dataset, with the total number of clusters (14) chosen so that in all experiments, if the proportionality of the colors is balanced, there can be at least two centers per color.
We set a timeout of 24 hours on each execution.

The algorithms are implemented using Java 1.8, with the code publicly available\footnote{\url{https://github.com/FraVisox/FKC}}.
The experiments have been carried out on a machine with 500Gb of RAM, equipped
with a 32-core AMD EPYC 7282 processor.

\paragraph{Performance metrics}

We consider four performance indicators:
the number of points maintained in memory by the algorithms;
the running times of both the \algo{Update} and \algo{Query} procedures;
and the approximation \emph{ratio}, namely the ratio between the obtained radius and the best radius ever found by \Chen or \Jones when run on all points of the window. All the metrics are computed as averages over 200 consecutive sliding windows of the stream.

\subsection{Dependency on the coreset size}\label{sec:delta-dependency}

A first set of experiments focuses on the influence of the coreset size on the performance of our algorithms. We fix the window size to 10\,000 points, and obtain varying  coreset sizes by 
varying $\delta$ between 0.5 and 4. The top graphs of
Figure~\ref{fig:deltas-ratio} show, for all datasets and all algorithms, how the 
solution quality, assessed through the approximation ratio defined above, 
changes with $\delta$.
We observe that \Ours and \OursObl find solutions of comparable quality, as do \Chen and \Jones (observe that the performance of these two latter algorithms is clearly independent of $\delta$, since they are always executed on the entire window). When
the coreset is smallest (i.e., $\delta=4$), clearly our algorithms find worse solutions, but always within a factor 2 from those returned by the baselines; as the coreset gets
larger (i.e., for smaller $\delta$), then the solutions become of quality
comparable (and, surprisingly, sometimes better) to the one computed by the sequential baselines.
\begin{figure}
    \centering
    \includegraphics[width=\linewidth]{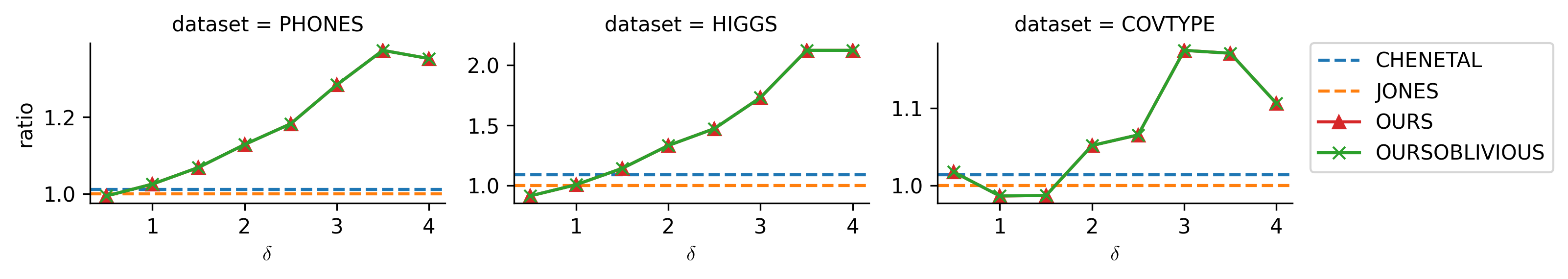}
    \includegraphics[width=\linewidth]{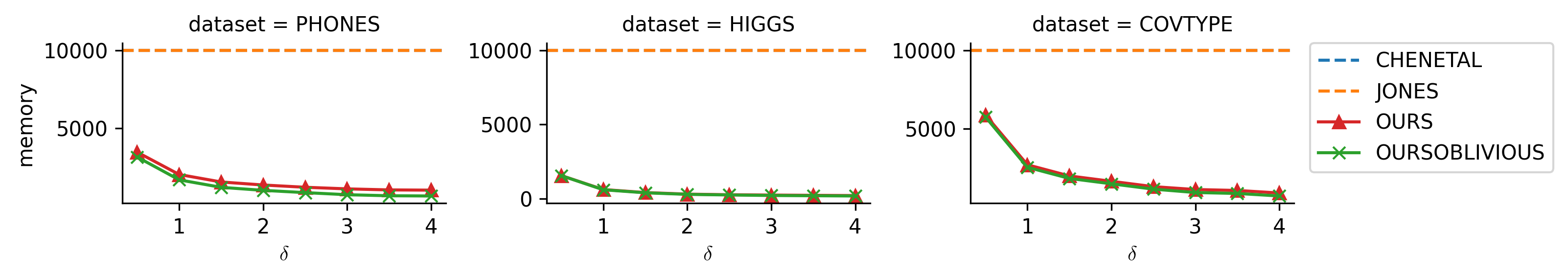}
    \caption{Approximation ratio (top) and memory (bottom, in number of points) for varying $\delta$.}
    \label{fig:deltas-ratio}
\end{figure}
The bottom graphs of Figure~\ref{fig:deltas-ratio} compare the memory usage of the
algorithms. As expected, for all values of $\delta$, our algorithm uses substantially less memory than the baselines, which maintain the entire window in memory, with gains becoming more marked as $\delta$ grows. This  behavior provides experimental evidence of the interesting memory-accuracy tradeoff featured by our algorithms, and is in accordance with the theoretical analysis. We also observe that \OursObl always requires (slightly) less memory than \Ours, which is due to the fact that maintaining updated estimates of $\dmin$ and $\dmax$ for the current window, rather than for the entire stream, reduces the space of feasible guesses, and thus the number of stored points. These space savings could become more significant for streams where the global aspect ratio is sensibly larger than the average aspect ratio within a window. 

\begin{figure}
    \centering
    \includegraphics[width=\linewidth]{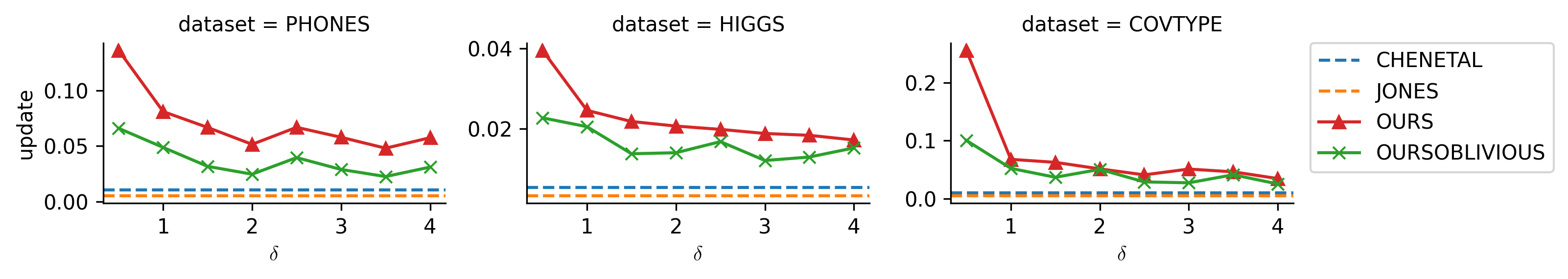}
    \includegraphics[width=\linewidth]{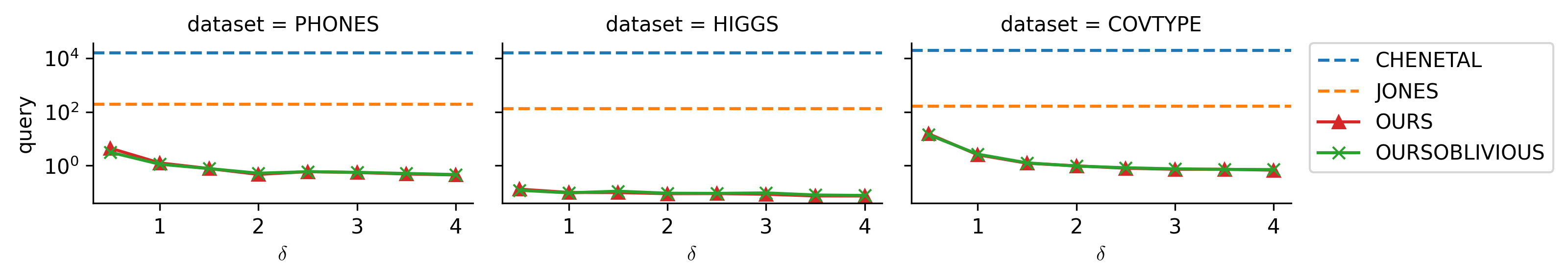}
    \caption{Running time in milliseconds of update (top) and query (bottom) for varying $\delta$. The scale of query time is logarithmic.}
    \label{fig:deltas-times}
\end{figure}

The top and bottom graphs of Figure~\ref{fig:deltas-times} report the runtime performance, in milliseconds,
for the update and query procedures, respectively.
Clearly, the sequential baselines feature next-to-zero update time, simply because they store the entire window, thus, at each time step, there is nothing to do other than adding one point and removing one to maintain the current window. Nevertheless, both \Ours{} and \OursObl feature reasonable update times, not exceeding 0.15 milliseconds, which is just a tiny fraction of the cost of a query.
It is interesting to observe that despite the 
overhead incurred by \OursObl to estimate  $\dmin$ and $\dmax$ for the current window, the advantage of having fewer guesses makes it always faster than \Ours{}. In both cases, using larger values of $\delta$ (i.e., smaller coresets) leads to faster update times, as expected. However, it is on query times where the difference with the baselines is most striking: by confining the execution of the expensive baselines only to a small coreset of window points, \Ours{} and \OursObl are able to find a solution for every window \emph{up to two orders of magnitude faster} than \Jones, which is in turn two orders of magnitude faster than \Chen. Even with the smallest values of $\delta$, 
which afford comparable accuracy  to the sequential baselines, the gains are between one and two orders of magnitude!

\subsection{Dependency on the window size}\label{sec:window-size-dependency}

The next set of experiments analyzes the behavior of the algorithms for varying window sizes.
For \Ours{} and \OursObl,   we consider the most accurate and slowest setting in the benchmark, 
thus fixing $\delta=0.5$.
(For the sake of space, we omit plots for the update time and the approximation ratio, which are fully consistent
with the findings of the previous section).
Figure~\ref{fig:wsize} (top) reports on the memory usage. Clearly, the memory used by the sequential baselines 
increases linearly with the window size. In contrast, 
the memory used by both versions of our algorithm, after an initial increase,  stabilizes to a value \emph{independent of the window size}.
This difference is reflected in the query time performance (Figure~\ref{fig:wsize}, bottom).
As observed in the previous section, the time gain when employing coresets compared to the use of the sequential baselines on the entire window is already of orders of magnitude for relatively small windows, and, as shown by the plots in the figure,
it keeps increasing very steeply with the window size.  In fact, for windows larger than 30\,000 points, the execution of the experiments relative to \Chen timed out, requiring about 400 seconds per window.
Similarly, those relative to \Jones timed out for windows larger than 200\,000 points.

\begin{figure}[t]
    \centering
    \includegraphics[width=\linewidth]{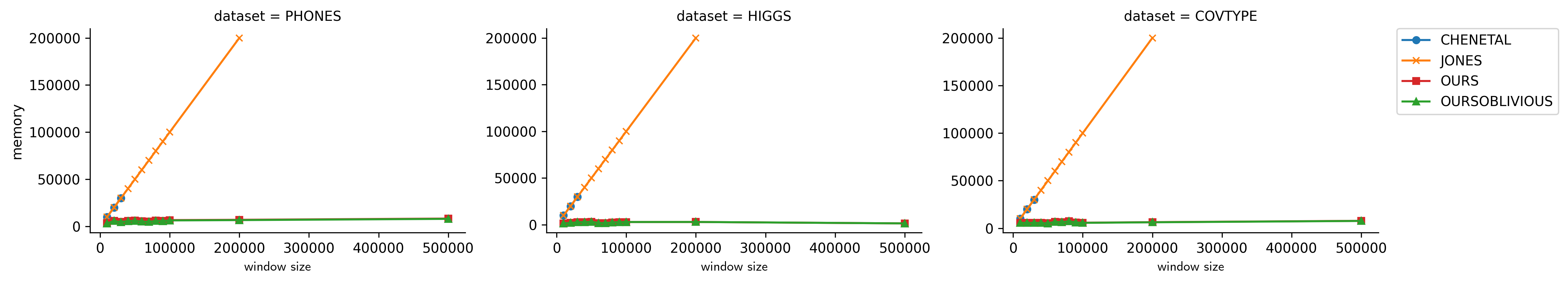}
    \includegraphics[width=\linewidth]{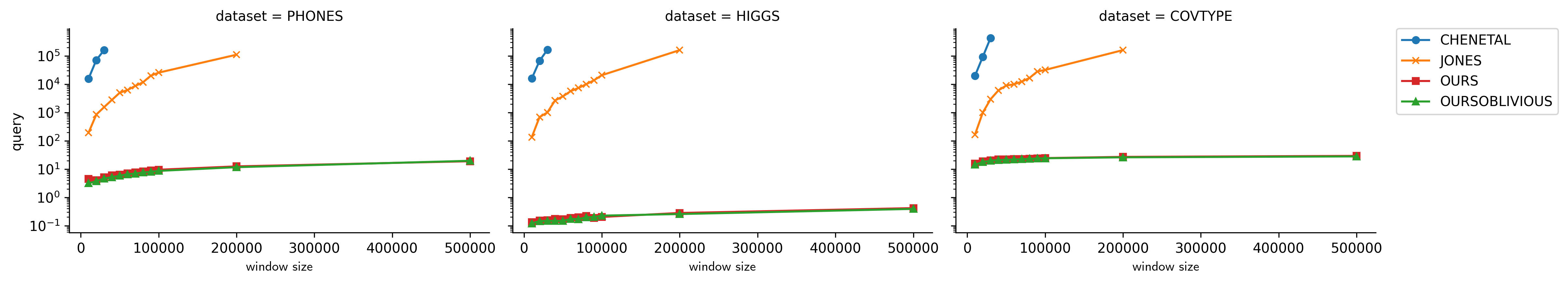}
    \caption{Memory (top, in number of points) and running time of query (bottom, in milliseconds) at varying window sizes. The scale of query time is logarithmic.}
    \label{fig:wsize}
\end{figure}

\subsection{Dependency on the dimensionality}\label{sec:dimensionality-dependency}

\begin{figure}[t]
\centering
  \includegraphics[width=.45\linewidth]{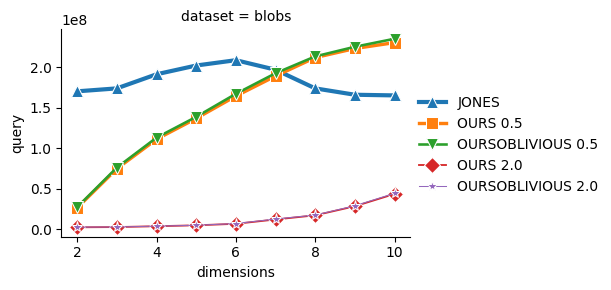}
  \includegraphics[width=.45\linewidth]{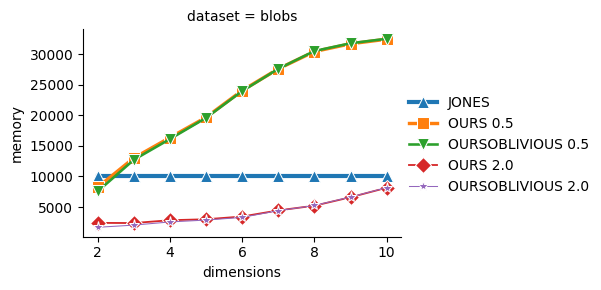}
    \caption{Query time (left, in milliseconds) and memory (right, in number of points) with respect to the dimensionality, on the \dataset{blobs} datasets.}
    \label{fig:dimensionality}
\end{figure}

\begin{figure}[t]
\centering
  \includegraphics[width=.45\linewidth]{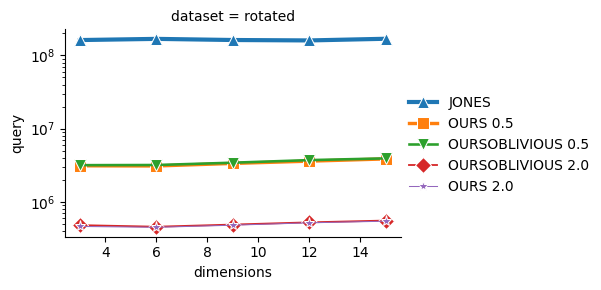}
  \includegraphics[width=.45\linewidth]{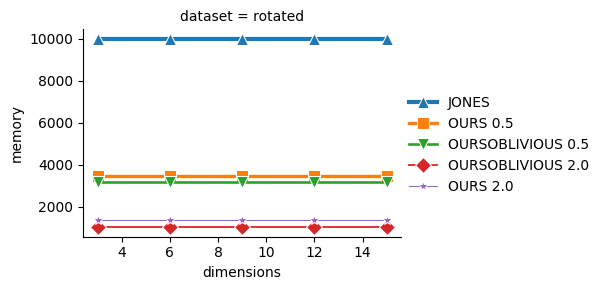}
    \caption{Query time and memory (resp. left and right) with respect to the dimensionality on the \dataset{rotated} datasets. The scale of query time is logarithmic.}
    \label{fig:rotation-experiments}
\end{figure}

This latter set of experiments studies the influence of the data dimensionality on the performance of our algorithm.
To this end we generated two families of synthetic datasets:
\dataset{blobs}, used to assess how the memory and query times of our algorithm are influenced by the dimensionality; and \dataset{rotated}, used to verify that these costs are related to the \emph{actual} dimensionality of the dataset, rather then the number of 
coordinates of each point.
The \dataset{blobs} datasets have 1\,000\,000 points: each dataset is a mixture of 21 multivariate $d$-dimensional Gaussians, for $2\le d\le 10$. The covariance matrix is $\Sigma=I_d\sigma^2$ with 
$\sigma=2$, and each point is assigned a random color out of 7, ensuring an even distribution. We set $k_i=3$, for $1\le i\le 7$, and the window size to $10\,000$.
The \dataset{rotated} datasets are derived from \dataset{PHONES} (which has 3 dimensions) by adding zero-filled dimensions to the original data, followed by a rigid random rotation of the entire extended dataset.
With this procedure we generate datasets with up to 15 nonzero coordinates, but by construction, all data still lie on a 3-dimensional subspace.

In these experiments, we use only \algo{JONES} as baseline and,
for our algorithm, we consider only two settings: 
$\delta=0.5$, which yields larger coresets but higher accuracy; 
and $\delta=2$, which yields smaller coresets and lower accuracy (that is still comparable, however, to the one attained by  the sequential baseline).
For conciseness, we omit the plots of the update time and the approximation ratio,
which are coherent with previous findings.

Figure~\ref{fig:dimensionality} reports the query time (left) and memory (right) on the \dataset{blobs} datasets. 
As expected, the performance of sequential baseline \algo{JONES} is insensitive to the dimensionality, while
both query time and memory usage of our algorithm grow with the dimensionality, 
 the growth being much steeper in the larger coreset setting ($\delta=0.5$), 
as suggested by the theoretical bounds.
We remark that our algorithm uses less memory than the sequential baseline for $\delta=2$, even for higher dimensions.

Figure~\ref{fig:rotation-experiments} reports results on the \dataset{rotated} datasets.
Note that in this case, where the data lies in a 3-dimensional subspace, the query time and memory for our algorithm are independent of the number of dimensions of the dataset. This confirms that our algorithm's performance depends on the \emph{actual} dimensionality of the dataset, rather than on the sheer number of coordinates of the vectors associated to the points.

\section{Conclusions} \label{sec-conclusions}
This paper presents the first sliding-window algorithm to enforce fairness in center selection, under the $k$-center objective. The algorithm works for general metrics, stores a number of points independent of the size of the window,  and provides approximation guarantees comparable to those of the best sequential algorithm applied to all the points of the window. Its space and time requirements are analyzed in terms of the dimensionality and of the aspect ratio of the input stream, although these values, which are difficult to estimate from the data, do not have to be known to the algorithm. A variant of the algorithm affording a  dimensionality-independent analysis, while still achieving $\BO{1}$ approximation, is also provided. Among the many avenues for future work, we wish to mention the extension of our algorithms to the robust variant of fair center, tolerating a fixed number of outliers, possibly fairly chosen among the most distant points w.r.t.\ a given solution. We believe that good approximations for robust fair center in sliding windows may be attained by  building on previous work for robust unconstrained $k$-center, matroid and fair center in the literature \citep{ChakrabartyN19,PellizzoniPP22a,DBLP:journals/jbd/CeccarelloPPS23,DBLP:conf/aistats/Amagata24}.

\backmatter

\section*{Declarations}

\subsubsection*{Funding}
This work was partially supported by MUR of Italy, under Projects PRIN 2022TS4Y3N - EXPAND, and PNRR CN00000013 (National Centre for HPC, Big Data and Quantum Computing).

\subsubsection*{Competing interests}
The authors have no relevant financial or non-financial interests to disclose.

\subsubsection*{Data availability}
The datasets are available at the UCI repository:
\dataset{PHONES} \url{https://doi.org/10.24432/C5689X},
\dataset{HIGGS} \url{https://doi.org/10.24432/C5V312}, and
\dataset{COVTYPE} \url{https://doi.org/10.24432/C50K5N}.

\subsubsection*{Code availability}
The code implementing the algorithms and the generated datasets of Section~\ref{sec-experiments} are publicly available at: 
\url{https://github.com/FraVisox/FKC}.

\subsubsection*{Author contribution}
All authors contributed to the study's conception, design, data collection, and analysis. All authors read and approved the final manuscript.

\bibliography{references}

\end{document}